\documentclass{elsart}

\input epsf
\usepackage{graphics, graphicx}
\usepackage{epsfig}
\usepackage{subfigure}
\usepackage{makeidx}
\usepackage{amsmath,amssymb}

\newtheorem{lemma}{Lemma}[section]
\newtheorem{theorem}[lemma]{Theorem}
\newtheorem{corollary}[lemma]{Corollary}

\begin{document}

\begin{frontmatter}



\title{Characterizing Graphs of Zonohedra}
\author{Muhammad Abdullah Adnan\corauthref{cor1}}, \ead{adnan@cse.buet.ac.bd}
\author{Masud Hasan}

\ead{masudhasan@cse.buet.ac.bd}

\corauth[cor1]{Corresponding author.}

\address{Department of Computer Science and Engineering,\\
Bangladesh University of Engineering and Technology (BUET), \\
Dhaka-1000, Bangladesh}


\begin{abstract}
A classic theorem by Steinitz states that a graph $G$ is
realizable by a convex polyhedron if and only if $G$ is
3-connected planar. Zonohedra are an important subclass of convex
polyhedra having the property that the faces of a zonohedron are
parallelograms and are in parallel pairs. In this paper we give
characterization of graphs of zonohedra. We also give a linear
time algorithm to recognize such a graph. In our quest for finding
the algorithm, we prove that in a zonohedron $P$ both the number
of zones and the number of faces in each zone is $O(\sqrt{n})$,
where $n$ is the number of vertices of $P$.
\end{abstract}

\end{frontmatter}

\bibliographystyle{abbrv}

\section{Introduction}
Polyhedra are fundamental geometric structures in 3D.
Polyhedra have fascinated mankind since prehistory.
They were first studied formally by the ancient Greeks
and continue to fascinate mathematicians, geometers,
architects, and computer scientists~\cite{HRZ97}.
Polyhedra are also studied in the field of art, ornament,
nature, cartography, and even in philosophy and literature~\cite{C97}.

Usually polyhedra are categorized based on certain mathematical
and geometric properties. For example, in platonic solids, which
are the most primitive convex polyhedra, vertices are incident to
the same number of identical regular polygons, in Archimedean
solids the vertices are allowed to have more than one type of
regular polygons but the sequence of the polygons around each
vertex are the same, etc. Among other types of polyhedra, Johnson
solids, prisms and antiprisms, zonohedra, Kepler-pointsod
polyhedra, symmetrohedra are few to mention. (See~\cite{KH01}, the
books~\cite{C73,C97}, and the webpages~\cite{H,MW} for more on
different classes of polyhedra.)

Polyhedra, in particular convex, provide a strong link between
computational geometry and graph theory, and a major credit for
establishing this link goes to Steinitz. In 1922, in a remarkable
theorem Steinitz stated that a graph is the naturally induced
graph of a convex polyhedron if and only if it is 3-connected
planar~\cite{g67,SR34}. Till today, Steinitz's theorem attracts
the scientists and mathematicians to work on it. For example,
there exist several proofs of Steinitz's
theorem~\cite{DG97,Zie95}.

One of the simplest subclasses of convex polyhedra are {\it
generalized zonohedra}, where every face has a parallel face and
the edges in each face are in parallel pairs~\cite{T92} (see
Figure~\ref{fig:fig11}(a)). This definition of zonohedra is
equivalent to the definition given originally by the Russian
crystallographer Fedorov~\cite{C73,T92}. Later Coxeter~\cite{C73}
considered two other definitions of zonohedra to mean more special
cases: (i) all faces are parallelograms and (ii) all faces are
rhombi. Coxeter called these two types as {\it zonohedra} and {\it
equilateral zonohedra} respectively. The polyhedra of
Figures~\ref{fig:fig11}(b) and ~\ref{fig:fig11}(c) are two
examples of these two types respectively. (For history and more
information on zonohedra, see~\cite{T92} and the web
pages~\cite{E,E96,H}.) In this paper, by {\it zonohedra} we mean
zonohedra defined by Coxeter.

\begin{figure}[!htbp]
\begin{center}
\includegraphics[width=\textwidth]{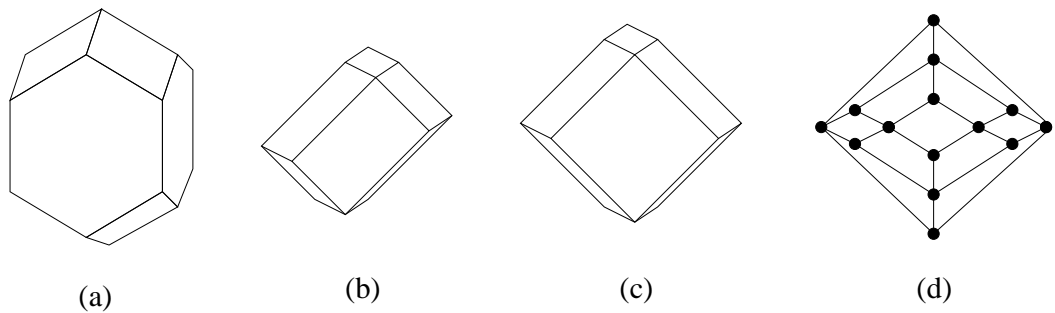}
\caption{(a) A generalized zonohedron, (b) a zonohedron, (c) an
equilateral zonohedron, and (d) the graph of the zonohedra of (b)
and (c).} \label{fig:fig11}
\end{center}
\end{figure}

As mentioned earlier, so far convex polyhedra have been classified
in many different classes based on geometric properties.
But to our knowledge they have not been classified based on their graphs.
Motivated by Steinitz's theorem, in this paper we characterize the
graphs of zonohedra (Section 3, 4). Graphs of zonohedra are also
called the {\it zonohedral graphs}. See
Figure~\ref{fig:fig11}(d). We also give a linear time algorithm
for recognizing a zonohedral graph (Section 5). As accompanying results,
we show that in a zonohedron $P$ both the number of ``zones'' and the number
of faces in each zone is $O(\sqrt{n})$,
where $n$ is the number of vertices of $P$
(Section $4$).

\section{Preliminaries}

A {\it graph} $G=(V,E)$ consists of a set of vertices $V(G)$ and a
set of edges $E(G)$ where an edge $(u,v)\in E(G)$ connects two
vertices $u,v\in V(G)$. A {\it path} in a graph $G$ is an alternating 
sequence of distinct vertices and edges where each
edge is incident to two vertices immediately preceding and
following it. A {\it cycle} is a closed path in $G$. The \emph{length} of 
a path (cycle) is the number of its edges.

A graph $G$ is {\it connected} if for two distinct vertices $u$
and $v$ there is a path between $u$ and $v$ in $G$. A (connected)
{\it component} of $G$ is a maximal connected subgraph of $G$. The
{\it connectivity} $\kappa (G)$ of $G$ is the minimum number of
vertices whose removal makes $G$ disconnected or a single vertex.
We say that $G$ is {\it $k$-connected} if $\kappa (G)\geq k$.
Alternatively, $G$ is $k$-connected if for any two vertices of $G$
there are at least $k$ disjoint paths.

A graph is {\it planar} if it can be embedded in the plane so that
no two edges intersect geometrically except at their common
extremity. A {\it plane graph} is a planar graph with a fixed
embedding in the plane. A planar graph divides the plane into
connected regions called {\it faces}.


A \emph{convex polyhedron} $P$ is the bounded intersection of a
finite number of half-spaces. The bounded planar surfaces of $P$
are called \emph{faces}. The faces meet along line segments,
called  \emph{edges}, and the edges meet at endpoints, called
\emph{vertices}.
By Euler's
theorem every convex polyhedron with $n$ vertices has
$\Theta(n)$ edges and $\Theta(n)$ faces~\cite{C97}. 
We call the number of edges belonging to a face
(of a polyhedron or a plane graph) as its {\it face length}.

Let $s$ be the unit sphere centered at the origin.
On $s$ the \emph{normal point} of a face $f$ of $P$ is the
intersection point of $s$ and the outward normal vector of $f$
drawn from the origin. In the {\it Gauss map} of $P$ on $s$ 
each face is represented by its normal point, each edge $e$ is
represented by the geodesic arc between the normal points of the
two faces adjacent to $e$, and each vertex $v$ is represented by
the spherical convex polygon formed by the arcs of the edges
incident to $v$.

In the {\it graph of a convex polyhedron} $P$, there is exactly
one vertex for each vertex of $P$ and two vertices are connected
by an edge if and only if the corresponding vertices in $P$ form
an edge.

Quite frequently we will use the same symbols for terms such as
edge, face, zone, etc. that are defined for both a polyhedron and
its graph.

\section{Characterization}

Faces of a zonohedron $P$ are grouped into different ``cycles'' of faces
called the \emph{zones} of $P$.
All faces of a zone $z$, which are called the \emph{zone faces}
of $z$, are parallel to a direction called the \emph{zone axis} of $z$.
For example, a cube has three zones with three zone axes perpendicular to each other.
Every zone face $f$ of $z$ has exactly two edges that are parallel to
the zone axis of $z$,
and the set of all such edges of the zone faces of $z$
are called the \emph{zone edges} of $z$.

In the Gauss map of $P$ a zone $z$ of $P$ is represented
as a great circle $g$ of $s$. Why? The normals of the
zone faces of $z$ are perpendicular to the zone axis of $z$. So
the corresponding normal points lie on a great circle whose plane
is perpendicular to the zone axis, and $g$ is exactly that great
circle. Observe that $g$ is the concatenation of small geodesic
arcs corresponding to the zone edges of $z$.

\begin{lemma}
\label{lem353} A zone $z$ of $P$ has even number of faces.
\end{lemma}


\begin{proof}
In $P$ any face $f$ has a parallel face $f'$.
Now consider the Gauss map of $P$ where $g$ is the
great circle corresponding to the zone $z$.
Since on $s$ the normal points of $f$ and $f'$ are antipodal,
if $g$ contains the nomral point of $f$ then it also contains
the normal point of $f'$.
Equivalently, if $f$ is in $z$, then $f'$ too is in $z$, and
thus the number of faces in $z$ is even. \qed
\end{proof}

\begin{lemma}
\label{lem355} Any two zones $z_1$ and $z_2$ of $P$ intersect into
two parallel faces.
\end{lemma}

\begin{proof}
Consider the Gauss map of $P$. Let $g_1$ and $g_2$ be the great
circles corresponding to the zones $z_1$ and $z_2$ respectively.
$g_1$ and $g_2$ intersect into two antipodal points $p_1$ and
$p_2$. Since a pair of antipodal points corresponds to a pair of
parallel faces, $z_1$ and $z_2$ intersect into nothing but the
pair of parallel faces corresponding to $p_1$ and $p_2$. \qed
\end{proof}

\begin{lemma}
\label{lem1} A face of $P$ belongs to exactly two zones of $P$.
\end{lemma}

\begin{proof}
A face of $P$ has two pairs of parallel edges which are zone edges
of two different zones.
\qed
\end{proof}


\begin{lemma}
\label{lem354} Let $f$ and $f'$ be two parallel faces of a zone
$z$. Then $f$ and $f'$ divides $z$ into two non-empty equal chains
of faces (excluding $f$ and $f'$) where each face in one chain has
its parallel pair in the other chain.
\end{lemma}

\begin{proof}
Consider the Gauss map of $P$. Let $g$ be the great circle
corresponding to the zone $z$. Since $f$ and $f'$ are parallel,
their normal points are two antipodal points $p$ and $p'$,
respectively, in $g$. $p$ and $p'$ divides $g$ into two half
circles $h_1$ and $h_2$. If $h_1$ contains a normal point
$p_1\notin \{p,p'\}$ of a face $f_1$ of $z$, then the antipodal
point of $p_1$, which is the normal point of the parallel face of
$f_1$, is in $h_2$.
\qed
\end{proof}

Let us now reflect the above (geometric) properties of $P$ to its
graph. Let $G$ be a 3-connected planar graph with even faces.
According to Tutte~\cite{T63,Wes01}, every 3-connected planar
graph has a unique planar embedding (and from now on by $G$ we
mean it with its unique planar embedding.) We override the
definition of a zone for $G$.
The pair of alternating edges of a (quadrilateral) face $f$ (of $P$ or $G$)
are called \emph{opposite} to each other in $f$.
A {\it zone} $z$ in $G$ is a cycle of \emph{zone faces} where for any three
consecutive zone faces $f_i,f_{i+1},f_{i+2}$ the common edge of $f_i$ and $f_{i+1}$
and the common edge of $f_{i+1}$ and $f_{i+2}$ are opposite in $f_{i+1}$,
and these common edges of the zone faces of $z$ are called the \emph{zone edges} of $z$.

As implied from the above properties of $P$, the necessary
conditions for $G$ to be a graph of a zonohedron are: $G$ is
3-connected planar, its face lengths are four, every face is a
zone face of exactly two zones, and any pair of zones intersect
into two faces and divide each other into two non-empty equal
chains of faces. Our main result consists in providing a
characterization by showing that the above conditions are also
sufficient.

\begin{theorem}
\label{thm311} A graph $G$ is a graph of a zonohedron iff $G$ is
3-connected planar, faces of $G$ are quadrilaterals, every face
 is a zone face of exactly two zones, and
each pair of zones in $G$ intersect into two faces and divide each
other into two non-empty equal chains of faces (excluding the two
common faces).
\end{theorem}

Note that the property of a 3-connected planar graph having a
unique planar embedding~\cite{T63} is used only for identifying faces of $G$.
It is not explicitly used in the characterization. This is because neither the number of zones nor
the number of faces in a zone depend on that unique embedding.

\section{Proof of the sufficiency}

\label{sufficiency} Our proof is constructive. The idea of our
construction is to delete zones one after another from the given
graph $G$ until we reach the graph of a cube. The graph of a cube
is the smallest graph satisfying the sufficient condition. Then
from a cube we construct a zonohedron by adding zones in reverse
order one after another.

The pair of faces in which two zones of $G$ intersect is called a {\it face pair}.
We define the {\it length} of a zone of $G$ as the number of faces in it.
Since two zones divide each other into two non-empty equal chains of faces,
the length of a zone is even and is at least four.

\subsection{Deleting zones from $G$}

We first define the deletion of a zone $z$. By {\it contraction}
of a zone edge $e=(u,v)$ we mean the replacement of $u$ and $v$
with a single vertex whose incident edges are the edges (other
than $e$) that were incident to $u$ or $v$.
Let $f$ be a zone face of $z$ and let $e_1$ and $e_2$ be two edges of $f$
that are zone edges of $z$.
We define the {\it contraction} of $f$ as follows.
We contract $e_1$ and $e_2$ into two vertices $w_1$ and $w_2$ respectively.
$w_1$ and $w_2$ now have two edges between them.
We replace these two edges by a single one (and
keep the other edges incident to $w_1$ and $w_2$ unchanged.)
We define the {\it deletion} of a zone $z$ as the contraction
of all zone faces of $z$.
See Figure~\ref{fig:fig72}.

\begin{figure}[!htbp]
\begin{center}
{\includegraphics[height=1.5in]{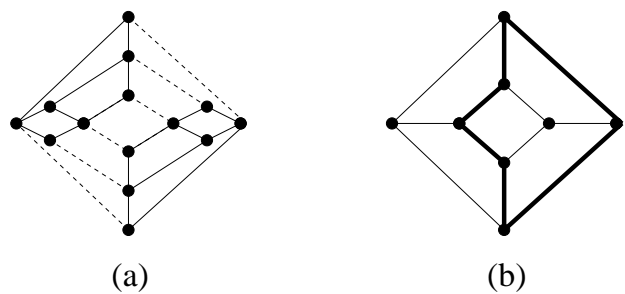}} \caption{(a)
Illustration of deletion of a zone from $G$, (b) the graph $G'$
after deletion of the zone. The heavily drawn edges show the cycle
resulting from the deleted zone.} \label{fig:fig72}
\end{center}
\end{figure}

Since a face $f$ belongs to exactly two zones, $z$ does not self
intersect. It implies that each contracted face of $z$ results
into an edge after the deletion. So as a whole, deleting $z$
results in a cycle whose length is same as that of $z$, and we call
this cycle as the \emph{zone cycle} corresponding to $z$.

We identify all the zones of $G$ as follows. For each face of $G$
by traversing the edges in circular order we
find the pair of opposite edges. 
Starting from a face $f$, we identify the two zones  that have $f$
in common. For each of them we group their adjacent faces one
after another based on opposite edges and check any intersection
within a zone before we come back to $f$. Similarly, we approach
other face pairs having opposite edges that have not been
encountered yet. Since every two zones intersect exactly two faces
the total number of times a face is traversed is at most half the
number of its edges. 

Once we identify the zones of $G$ we delete them one by one until
we reach the graph of a cube, for which the following observation
is obvious.

\begin{lemma}
\label{le:cube}
$G$ is the graph of a cube iff $G$ has three zones of length four.
\end{lemma}


Due to the above lemma we only delete the zones of size six or
more to reach the graph of a cube. The following lemma will prove
that we can successfully do that.


\begin{lemma}
\label{lem71} Let $G$ be a graph of a zonohedron with more than
three zones. Let $G'$ be the graph after deleting a zone $z$ from
$G$. Then $G'$ satisfies the conditions of Theorem~\ref{thm311}.
\end{lemma}

\begin{proof}
%
%
%

To prove that $G'$ is 3-connected, we first show that the new
vertex $w$ in $G'$ obtained by contracting an edge $e=(u_1,u_2)$
of $z$ has degree three or more. Suppose for a contradiction that
$d(w)<3$. Since $d(u_1)\geq 3$, $d(u_2)\geq 3$ then $d(w)\nless 2$
. If $d(w)=2$ then $d(u_1)=d(u_2)=3$. Suppose $u_1$ and $u_2$
belong to the faces $f_1$ and $f_2$ other than the
 faces $f_{z}$ and $f'_{z}$ of $z$ in $G$. If $f_1$ (similarly $f_2$) belongs to zones
 $z_1$ and $z_2$, then $f_2$ ($f_1$) also belongs to $z_1$ and
 $z_2$ as illustrated in Figure~\ref{fig:fignew1}. According to
 Theorem~\ref{thm311}, $z_1$ and $z_2$ divide each other into two
 non-empty equal chains of faces, each having exactly
 one face. Since each pair of zones intersect each other into two
 faces the number of zones in $G$ cannot be greater than three, a
 contradiction to the assumption that $G$ contains more than three
 zones.

\begin{figure}[!htbp]
\begin{center}
{\includegraphics[height=1.5in]{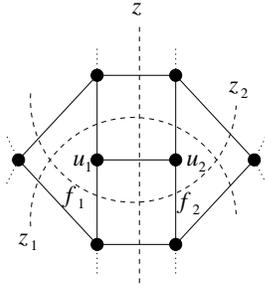}} \caption{A
zonohedral graph $G$ having $d(u_1)=d(u_2)=3$ (dashed lines
represent the zones).} \label{fig:fignew1}
\end{center}
\end{figure}

We now show that, there exist at least three disjoint paths
between any two vertices in $G'$ .In $G$ any two vertices $u$ and
$v$ had at least three disjoint paths and any contracted edge of
$z$ can be in at most one of these paths. Now in $G'$, if $u$ and
$v$ are contracted together, then we are done. If one of them, say
$u$, was contracted with its neighbor to a new vertex $w$
(similarly if none of $u$ and $v$ was contracted), then $w$
(similarly $u$) maintains those three disjoint paths to $v$,
possibly with smaller path lengths.

By deleting $z$ we have neither modified the faces of $G$ nor
introduced new faces. Hence if all the faces of $G$ belongs to
exactly two zones then every face of $G'$ belongs to exactly two
zones.

Next we show that any two zones in $G'$ divide each other into two
non-empty equal chains of faces. Let $z_1$ and $z_2$ be two zones
other than $z$ in $G$. Let $(f,f')$ be the face pair of $G$ at
which $z_1$ and $z_2$ intersect each other. Let the two equal
chains of faces of $z_1$ between $f$ and $f'$ be $l_1$ and $l_2$.
We will show that in $G'$, $l_1$ and $l_2$ have equal length of at
least two. A similar argument holds for $z_2$ and allows to
complete the proof.

Consider the intersection of $z$ and $z_1$. Let $(f_1,f_2)$ be the
face pair of $G$ at which $z$ and $z_1$ intersect. Since
$(f_1,f_2)$ divides $z_1$ into two other equal chains of faces,
$f_1$ is in $l_1$ (similarly in $l_2$) if and only if $f_2$ is in
$l_2$ (similarly in $l_1$). W.l.o.g. assume that $f_1$ is in
$l_1$. After deleting $z$, $z_1$ loses exactly two faces: $f_1$
from $l_1$ and $f_2$ from $l_2$. So in $G'$, $l_1$ and $l_2$ are
of equal length. Moreover, by Lemma~\ref{le:cube}, in $G$, $z_1$
has length at least six. So in $G'$, $l_1$ and $l_2$ have length
at least two.
\qed
\end{proof}

\subsection{Adding 3D zones}

Let the current zonohedron be $P'$ and its graph be $G'$.
Let $G$ be the graph from which $G'$ was obtained by deleting a zone.
We will add to $P'$ a zone $z$ corresponding to the deleted zone of $G$ as follows.

Let $P$ be the resulting polyhedron after adding $z$ to $P'$.
Let $c$ be the cycle (of edges) in $P'$ that corresponds to the zone cycle
of the deleted zone of $G$.
By Lemmas \ref{lem71} and \ref{lem354}, $c$ divides each zone of $P'$
into two equal chains of faces where faces in one chain have parallel pairs in
the other chain. As a whole, $c$ divides the set of faces of $P'$
into two subsets $P_1'$ and $P_2'$ where faces in $P_1'$ ($P_2'$) have
parallel pairs in $P_2'$ ($P_1'$).

To get $P$ we expand each edge of $c$ to a rhombus in a common
direction $d$. Clearly the graph of $P$ is $G$. See
Figure~\ref{fig:figcube}. What remain to be proven are: (i)~the
faces of $z$ are in parallel pairs, and (ii) there exists a $d$
such that $P$ is convex. We prove them in the following two lemmas
respectively.

\begin{figure}[!htbp]
\begin{center}
{
\input{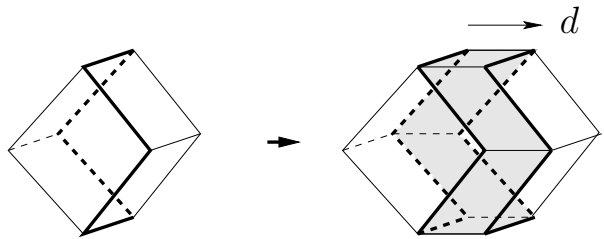}}
\caption{Adding a pseudo prism to a cube. Heavily drawn lines show $c$
and shaded faces are the newly added zone.
}
\label{fig:figcube}
\end{center}
\end{figure}

\begin{lemma}
\label{le:in_parallel} The faces of $z$ are in parallel pairs.
\end{lemma}

\begin{proof}
It suffices to show that edges of $c$ are in parallel pairs.
Consider an edge $e$ of $c$. Let $z$ be the zone of which $e$ is a
zone edge.
By Lemma~\ref{lem71}, $z$ crosses $c$ twice.
Let $e'$ be the other edge of $c$ at which $z$ crosses $c$. Since
$e$ and $e'$ belong to the same zone $z$, they are parallel (to
the zone axis of $z$). \qed \end{proof}


\begin{lemma}
\label{le:d}
There exists $d$ such that $P$ is convex. Moreover,
$d$ can be found in $O(h\log h)$ time where $h$ is the number of
faces of $z$.
\end{lemma}

\begin{proof}
To prove that $P$ is convex it suffices to prove that there exists
$d$ such that no face of $P'$ is parallel to $d$ and viewing $P'$
orthogonally from $d$ keeps $c$ as the boundary of the projection
(and thus makes all the faces in one side of $c$ visible and the
faces in other side invisible). We prove this using an induction
on the number of zones of $P'$.

For the basis of the induction we consider the smallest zonohedron
$P'$ (which is a cube) with three zones. Clearly, in a cube there
are four possible $c$ each of which divides the faces of the cube
into two sets of faces $P_1'$ and $P_2'$ where faces in $P_1'$
($P_2'$) have parallel pairs in $P_2'$ ($P_1'$). For each such $c$
there exists a $d$ where exactly the faces of $P_1'$ ($P_2'$) are
visible and the faces of $P_2'$ ($P_1'$) are invisible. Moreover,
$d$ is the resultant vector of the outer-normals of those visible
faces and is not parallel to any face of the cube.

Let $P''$ be the zonohedron whose zone cycle $c'$ was expanded in
direction $d'$ to obtain $P'$. Let $z'$ be the zone added to $P'$
due to this expansion (see Figure~\ref{fig:figd}(b)).

\begin{figure}[!htbp]
\begin{center}
{
\input{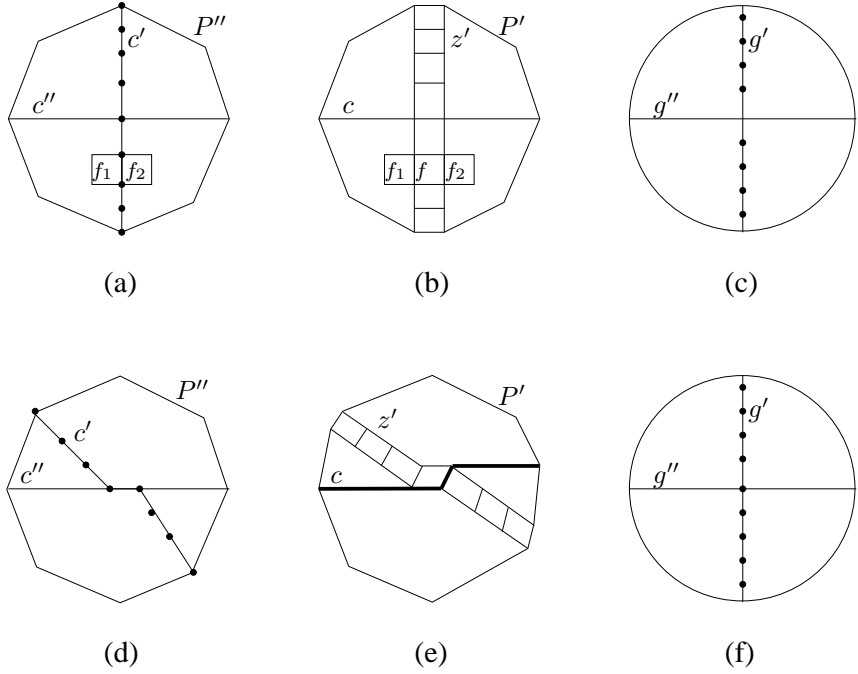}}
\caption{(a) The two zone cycles $c''$ and $c'$ in $P''$, (b) The
corresponding zone $z'$ and cycle $c$ in $P'$, (c) The great
circles $g''$ and $g'$ representing $c''$ and $c'$ in the Gauss
map (for Case 1), (d) The two cycles $c''$ and $c'$ sharing two
edges $e, e'$ in $P''$, (e) The corresponding zone $z'$ and cycle
$c$ in $P'$ and (f) The great circles $g''$ and $g'$ representing
$c''$ and $c'$ in the Gauss map (for Case 2). } \label{fig:figd}
\end{center}
\end{figure}

Remember that in $P'$, $c$ divides $z'$ into two chains of faces
where faces in one chain have parallel pairs in the other chain.
So $c$ contains two zone edges of $z'$. Let $e_1$ and $e_2$ be
those two edges. Moreover, since $c$ divides the faces of $P'$
into two sets $P_1'$ and $P_2'$ where faces in $P_1'$ ($P_2'$)
have parallel pairs in $P_2'$ ($P_1'$), the faces of $P''$
(without $z'$) are also divided by the cycle $c''=c'\setminus
\{e_1,e_2\}$ into two subsets of faces where faces in one set have
parallel pairs in the other set. Therefore, by induction
hypothesis there is a direction $d''$ which is not parallel to any
face of $P''$ and from which all faces in one side of $c''$ are
visible and all faces in the other side of $c''$ are invisible.
Let the set of visible and invisible faces be $P_1''$ and $P_2''$
respectively.


For the remaining proof we will switch our attention to the Gauss
map.
Let $g''$ be the great circle whose plane is perpendicular to
$d''$. Let the two half spheres defined by $g''$ be $h_1$ and
$h_2$. Assume that $h_1$ ($h_2$) is visible (invisible) to $d''$.
So the normal-points of the faces of $P_1''$ and $P_2''$ are
within $h_1$ and $h_2$ respectively.

In $P''$, $c''$ and $c'$ must intersect (Figure~\ref{fig:figd}(a))
possibly sharing some edges (Figure~\ref{fig:figd}(d)). Hence we
have two cases.


{\bf Case 1:} $c''$ and $c'$ intersect in a pair of vertices.

For this case we will prove that $d''$ will work as $d$. Let $f_1,
f_2$ be two arbitrary adjacent faces of $P''$ whose common edge
$e$ is in $c'$. Assume that the normal points of $f_1$ and $f_2$
are in $h_1$ (similarly in $h_2$). After expanding $P''$ to $P'$,
let the face created from $e$ be $f$. Since by expansion the
normal point of $f_1$ and $f_2$ remain unchanged, it suffices to
prove that the normal point of $f$ is also in $h_1(h_2)$ (see
Figure~\ref{fig:figd}(b)). Since $f_1,f,f_2$ are three adjacent
faces of a zone (other than $z'$) of $P'$, their normal points
must lie on a great circle and the normal-point of $f$ is in the
geodesic arc connecting that of $f_1$ and $f_2$. Therefore, the
normal-point of $f$ must be within $h_1$($h_2$).





{\bf Case 2:} $c''$ and $c'$ share some edges.

We will first prove that $c''$ and $c'$ share exactly two  edges.
Let the two great circles of $d''$ and $d'$ be $g''$ and $g'$
respectively. Edges/vertices of $c''$ ($c'$) represent points/arcs
of $g''$ ($g'$) respectively. (For $c'$ simply think its
edges/vertices as the zone faces/zone edges of $z'$ and for $c''$
simply think the edges/vertices of $c''$ as the zone faces/zone
edges of the zone that would be created if $c''$ were expanded in
direction $d''$). Now, $g''$ and $g'$ intersect into two antipodal
points and their corresponding two edges are only common in $c''$
and $c'$. See Figure~\ref{fig:figd}(d,e,f).

After creating $z'$ the two common edges become two parallel
faces. Let they be $f$ and $f'$. By the argument of Case 1, except
$f$ and $f'$ all faces in one side of $c$ are visible and all
faces in the other side are invisible from $d''$. If $f$ and $f'$
too are not parallel to $d''$ and are visible/invisible as
required, then $d=d''$, and we are done. Otherwise, we can always
take $d$ as $d''+\epsilon$, where $\epsilon$ is small enough such
that $f$ and $f'$ are no more parallel to $d$, they become
visible/invisible as required, and the visibility/invisibility of
all other faces remain the same. Note that there may be one more
case: it may be possible that $f(f')$ is supposed to be visible
(invisible) but is invisible (visible) from $d''$. Then by
symmetry of $P'$ we can simply interchange $f$ and $f'$ in $P_1'$
and $P_2'$ and thus take $d=d''$.

Now $d$ can be easily found as follows. From the above argument it
is clear that there exists a $d$ from which all the faces in one
side of $c$ are visible. In fact $d$ is a direction in the
intersection of the positive half-spaces (the positive half-space
of a face $f$ is the plane of $f$ from which $f$ is visible) of
the faces of $P_1$ adjacent to $c$. Hence determining $d$ takes
$O(h\log h)$ time~\cite{BHL04} where $h$ is the number of edges of
$c$. \qed
\end{proof}

\subsection{Running Time}
Now we examine the time complexity of the construction as a whole.
Finding the zones and the face pairs take linear time. Deletion of
zones also takes linear time. Final points of $P$ are calculated
by the amount of expansion of all zones of $G$ by
Lemma~\ref{le:d}. Thus the total time required for all expansion
is $O(h_1\log h_1 + h_2\log h_2 + \cdots + h_m\log h_m)$, where
$m$ is the number of zones of $G$ and $h_i$ is the number of faces
of the new zone at $i$-th expansion. Since all the faces are
quadrilaterals, at $i$-th step, $h_i$ new faces are created. Hence
the sum $h_1+h_2+\cdots+h_m$ is the total number of faces which is
$O(n)$, where $n$ is the number of vertices in $G$. Moreover, $h_1
= 4$ and from Lemma~\ref{lem355}, $h_i=h_{i-1}+2$, which implies
that $m=O(\sqrt{n})$.

\begin{theorem}
The number of zones in a zonohedron is $O(\sqrt{n})$.
\end{theorem}

\begin{cor}
The maximum number of faces in a zone is $O(\sqrt{n})$.
\end{cor}

\begin{proof}
By Lemma~\ref{lem355}, every two zones intersects into two
parallel faces. So a zone can intersect $O(\sqrt{n})$ other zones
in $O(\sqrt{n})$ faces. \qed
\end{proof}

Therefore, the total running time of the construction algorithm is
$O(h_m\log (h_1 \cdot h_2 \cdot \cdots \cdot h_m))$ =
$O(\sqrt{n}\log(4\cdot 6 \cdot 8 \cdot \cdots \cdot (4+2m)))$ =
$O(\sqrt{n}\log(2^m(m+2)!))$ = $O(\sqrt{n}(m\log 2+m\log m))$ =
$O(\sqrt{n}(\sqrt{n}+\sqrt{n}\log n))$ = $O(n\log n)$.


\begin{theorem}
\label{thm322} A zonohedron $P$ from a zonohedral graph $G$ can be
constructed in $O(n\log n)$ time, where $n$ is the number of
vertices in $G$.
\end{theorem}

\section{Recognizing a zonohedral graph}

\label{section:recog} Let $G$ be the given graph. $G$ can be
tested for 3-connected planar in linear time~\cite{NC88}. Testing
whether all faces of $G$ are even takes linear time. We already
discussed that finding zones and face pairs takes linear time.
Once the face pairs are determined, we can measure how a zone is
divided by its face pairs. For all zones it takes linear time in
total.

\begin{theorem}
\label{thm411} Given a graph $G$, recognizing whether $G$ is zonohedral
can be done in linear time.
\end{theorem}

Observe that our recognition of a zonohedral graph will also work
for recognizing the graph of a generalized zonohedron.

\begin{cor}
The graph of a generalized zonohedron can be recognized in linear
time.
\end{cor}

\section{Conclusion}

An immediate open problem is to characterize graphs of other
subclasses of convex polyhedra, in particular graphs of
generalized zonohedra. A generalized zonohedron contains faces of
length greater than four. The difficulty with characterizing
graphs of generalized zonohedra is that after deletion of a zone
the cycle $c$ may contain faces. Hence during the construction we
have to prove that those faces are in parallel pairs and a great
circle exists through them.

Our construction of $P$ starts with a cube.
But it will also work if we started with a parallellopiped.



\bibliography{zonobib}

\end{document}